\documentclass[a4paper,11pt]{article}
\usepackage[english]{babel}
\usepackage[utf8]{inputenc}
\usepackage[T1]{fontenc}
\usepackage{lmodern}

\usepackage{amsmath,amsfonts,amsthm,amssymb}

\usepackage{subcaption}
\usepackage{graphicx}
\graphicspath{{./fig/}}

\usepackage{xcolor}
\usepackage{changepage} 

\usepackage{hyperref}
\usepackage[backend=biber, style=alphabetic, sorting=nymt, maxbibnames=5]{biblatex} 
\DeclareSortingTemplate{nymt}{ 
  \sort{
    \field{presort}
  }
  \sort[final]{
    \field{sortkey}
  }
  \sort{
    \field{sortname}
    \field{author}
  }
  \sort{
    \field{sortyear}
    \field{year}
  }
  \sort{
    \field[padside=left,padwidth=2,padchar=0]{month}
    \literal{00}
  }
  \sort{
    \field{sorttitle}
 }
}
\newtheorem{theorem}{Theorem}
\newtheorem{definition}{Definition}
\newtheorem{proposition}{Proposition}
\newtheorem{lemma}{Lemma}
\newtheorem{corollary}{Corollary}

\newtheorem{remark}{Remark}

\newtheorem{claim}{Claim}
\newtheorem{conjecture}{Conjecture}

\newtheorem{question}{Question}

\newcommand{\tiling}{\mathcal{T}}
\newcommand{\polygon}{\rho} 
\newcommand{\chain}{\chi}
\newcommand{\epath}{\mathbf{p}} 
\newcommand{\edge}{\vec{e}}

\newcommand{\patch}{P}

\newcommand{\adj}{\sim} 
\newcommand{\vadj}{\sim_{\mathrm{v}}} 
\newcommand{\neighbour}[1]{\mathcal{N}(#1)} 
\newcommand{\vneighbour}[1]{\mathcal{N}_{\mathrm{v}}(#1)} 

\newcommand\sites{V} 
\newcommand\configurations{X}

\newcommand\bootstrapCA{\ensuremath{F}}

\newcommand{\invasion}{I}
\newcommand{\chainzero}{C}
\newcommand{\blocking}{B}
\newcommand{\enclosed}[1]{E_{#1}}
\newcommand{\pcriticalzero}{p_c^0}
\newcommand{\pcriticalone}{p_c^1}
\newcommand{\MNVPC}{\texttt{MNVPC}}

\newcommand{\ZZ}{\mathbb{Z}}
\newcommand{\NN}{\mathbb{N}}
\newcommand{\RR}{\mathbb{R}}
\renewcommand{\epsilon}{\varepsilon}

\graphicspath{{./figures/}}

\title{Bootstrap percolation on rhombus tilings}

\date{2024}

\author{S. Esnay \and V. Lutfalla \and G. Theyssier}

\begin{document}

\maketitle

\begin{abstract}
  2-boostrap percolation on a graph is a diffusion process where a vertex gets infected whenever it has at least 2 infected neighbours, and then stays infected forever.
  It has been much studied on the infinite grid for random Bernoulli initial configurations, starting from the seminal result of van Enter that establishes that the entire grid gets almost surely entirely infected for any non-trivial initial probability of infection.

  In this paper, we generalize this result to any adjacency graph of any rhombus tiling of the plane, including aperiodic ones like Penrose tilings. We actually show almost sure infection of the entire graph for a larger class of measure than non-trivial Bernoulli ones.

  Our proof strategy combines a geometry toolkit for infected clusters based on chain-convexity, and uniform probabilistic bounds on particular geometric patterns that play the role of 0-1 laws or ergodicity, which are not available in our settings due to the lack of symmetry of the graph considered.
\end{abstract}

\section{Introduction}
\label{sec:intro}

Bootstrap percolation was first introduced by \cite{chalupa1979} as a simplified model of a magnetic system progressively loosing its magnetic order\footnote{Due to initial non-magnetic impurities and competition between exchange interactions (which tends to align atomic spins) and crystal fields interactions (which tends to force atoms into a singlet state, and thereby suppressing magnetic moments).}.
When abstracted away from this physical modeling, it can be seen as a simple growing or diffusion process on a graph: at each time step, any vertex with $m$ or more \emph{infected} neighbours becomes infected, and infected vertices remain infected forever.
As in classical percolation theory, this system is usually studied on random initial condition (each vertex is initially infected with some probability $p$), and the main question is to determine whether the entire graph will be infected almost surely (depending on $p$).

This process has been first studied on the bi-dimensional grid where the case ${m=2}$ is the most interesting one. 
\cite{vanenter1987} proved in that case that whenever ${p>0}$, the entire graph is infected almost surely. 
Then \cite{aizenman1988} and later \cite{holroyd2003} obtained sharp thresholds on $p$ for finite grids as a function of their size.
Still on the grid (in dimension two or higher), these seminal results were considerably generalized to a large class of monotone cellular automata called U-bootstrap and introduced in \cite{bollobas2015} (see \cite{morris2017} for a detailed survey of this family and \cite{bollobas2023,balister2023} for a taste of the most recent results).

On the other hand, bootstrap percolation was studied on many different graphs: finite ones like in \cite{gravner2017,gravner2015}, tilings by regular polygons \cite{bushaw2018} where critical probabilities are trivial ($0$ or $1$) like on the infinite grid, or hyperbolic lattices \cite{sausset2009} or regular trees and Cayley graphs of non-amenable groups \cite{biskup2009,balogh2006} where non-trivial critical probabilities were established.
A common aspect to these works (and to most of the literature to our knowledge) is that the graphs considered are highly symmetric which gives, on one hand, powerful probabilistic tools for tackling the problem (0-1-laws, ergodicity, etc), and, on the other hand, a nicer and more uniform geometry to infected clusters.

There is a general observation (or belief) in statistical mechanics that many features of studied systems only depend on a few parameters and should be independent of the details of the system.
This is often referred to as universality, but there is no rigorous mathematical definition of the concept to our knowledge. More modestly and back to the precise context of bootstrap percolation, one could ask how robust is the seminal result of \cite{vanenter1987} when modifying the underlying graph. First, as mentioned above, there is a dramatic change of behavior when going from the grid to trees, Cayley graphs of non-amenable groups or hyperbolic graphs (non-critical probabilities arise).
But what is special about the regular grid $\ZZ^2$ among planar graphs that are quasi-isometric to it?
what is the importance of symmetry (in particular translation invariance) that allows to make concise proofs?

In this paper, we consider boostrap percolation on any adjacency graph of any rhombus tiling  of the plane with finitely many tile types up to translation, including the regular grid but also highly non-symmetric ones like Penrose tilings.
Our main result is a generalization of the seminal result of \cite{vanenter1987} to this entire class of graphs, precisely: on the adjacency graph of any rhombus tiling, for any non-zero initial probability of infection, the entire graph will get infected almost surely.

Classical percolation was already studied on Penrose tilings in \cite{hof1998}, where the lack of translation invariance and symmetry is tackled by introducing an ergodic measure that averages over all possible Penrose tilings. Our approach is different and gives results for bootstrap percolation on any single tiling without averaging. We generalize Enter's argument to this setting by combining a geometrical analysis of blocked infected cluster based on chain-convexity (which are simply rectangles in the case of the grid) and uniform probability upper bounds on blocked cluster events that bypass the lack of translation invariance (whereas thanks to ergodicity in the case of the grid, it is enough to bound the probability of a blocked infected cluster containing the origin).

\paragraph*{Contents of the paper.} In section \ref{sec:bootstrap-intro}, we introduce bootstrap percolation cellular automata on any graph, the main problem, and some basic probabilistic notions and ingredients to be used later. In section~\ref{sec:rhombus}, we introduce rhombus tiling and some of their geometric features. The two previous objects introduced separately meet in section~\ref{sec:mainresult}, where we establish our main result in two steps: first, by a detailed analysis of infection clusters based on chain-convexity (sub-section \ref{sec:geometrical_elements}), and then, by a probabilistic argument applied to specific geometric objects previously identified (sub-section~\ref{sec:probargs}).
In the appendices we present two examples that limit the generalization of our result to other graphs and percolation processes.
In Appendix \ref{appendix:nonzeroone} we show that 2-neighbour percolation does not have a critical threshold (or $0-1$) behaviour on the adjacency graphs of arbitrary quadrilater tilings.
In Appendix \ref{appendix:perco-rhombus} we show that arbitrary percolation processes (here a variant of oriented bootstrap percolation) do not have a critical threshold behaviour on the adjacency graphs of arbitrary rhombus tilings.

\section{Bootstrap cellular automaton and probability measures}
\label{sec:bootstrap-intro}

In this section, we introduce the bootstrap cellular automaton on an arbitrary abstract graph and present basic facts around probability measures (to be used later in Section~\ref{sec:probargs}).

Given a graph $G=(\sites,E)$, a \emph{configuration} is a map giving a value $0$ or $1$ to each vertex.
We denote by ${\configurations = \{0,1\}^\sites}$ the set of configurations.
If $c\in\configurations$ and ${v\in\sites}$, we use notations $c(v)$ and $c_v$ interchangeably.
The set $X$ can be endowed with the pro-discrete topology (product of the discrete topology over ${\{0,1\}}$) which is generated by the collection of cylinder sets for $D\subseteq\sites$ a finite domain of $\sites$ and ${P:D\to\{0,1\}}$ a partial configuration of domain $D$:
\[[P] = \{c\in\configurations \mid \forall t\in D, c(t)=P(t)\}.\]
$X$ endowed with this topology is compact.
We will often use the notation shortcut $q^D$ for ${q\in\{0,1\}}$ and ${D\subseteq\sites}$ to denote the constant partial configuration equal to $q$ on domain $D$. 
Therefore ${[q^D]}$ denotes the set of configurations in state $q$ on domain $D$.

In this paper, we are interested in a dynamical process acting on configurations which is a particular cellular automaton often called \emph{bootstrap percolation} \cite{chalupa1979,vanenter1987,balogh2006}.

The $m$-neighbour contamination or $m$-bootstrap cellular automaton ${F:\configurations\to\configurations}$ for some integer $m$ is defined as follows:
\[F(c)_v =
  \begin{cases}
    1 &\text{ if $c_v=1$ or $\bigl|\{v' \mid (v',v)\in E\text{ and }c_{v'}=1\}\bigr|\geq m$},\\
    0 &\text{ else,}
  \end{cases}
\]
for any configuration ${c\in X}$ and any vertex ${v\in\sites}$.

This cellular automaton is both freezing and monotone (see \cite{salo2022}), which means that:
\begin{enumerate}
\item ${F(c)_v\geq c_v}$ for any ${c\in\configurations}$ and ${v\in\sites}$ (freezing),
\item ${F(c')_v\geq F(c)_v}$ for all ${v\in\sites}$ whenever ${c'_v\geq c_v}$ for all ${v\in\sites}$ (monotone).
\end{enumerate}

The freezing property lets us see $F$ as a growing process: starting from any initial configuration seen as a set of selected vertices of value $1$, $F$ can only make this set increase with time.
In particular, for any ${c\in\configurations}$ and any ${v\in\sites}$, the sequence ${\bigl(F^n(c)_v\bigr)_n}$ is ultimately constant of limit value $c^\infty_v$. Equivalently, the sequence of configurations ${\bigl(F^n(c)\bigr)_n}$ always converges and its limit is precisely $c^\infty$.

We will focus on the set $\invasion$ of initial configurations that invade the entire graph during this growing process, formally:
\[\invasion = \bigl\{x\in\configurations \mid \forall v\in\sites, \exists n\in\NN, F^n(c)_v = 1\bigr\}.\]

Clearly not all configurations lie in $\invasion$.
We want to understand how big $\invasion$ is.
Topologically, it is a $G_\delta$ set (intersection over ${v\in\sites}$ of sets ${\cup_n\{c \mid F^n(c)_v=1\}}$ which are open), but following motivations from statistical physics and percolation theory, we will study $\invasion$ through probability measures.

A Borel probability measure is uniquely determined by its value on cylinders (by Carathéodory-Fréchet extension theorem, see \cite[Theorem 0.5]{walters1981}).
We mainly consider \emph{Bernoulli measures} $\mu_p$ which are product measures determined by a parameter $p\in[0,1]$ giving the ``probability of having a $1$ at one vertex'', as follows:
\[\mu_p([P]) = \prod_{v\in D, P(v)=1} p \ \ \times\prod_{v\in D, P(v)=0}(1-p).\]

The monotony property of $F$ mentioned above implies that the measure of $I$ under Bernoulli measures can only increase with $p$: ${\mu_p(I)\leq \mu_{p'}(I)}$ whenever ${p\leq p'}$ (this can be shown by a straightforward coupling argument).
Since ${\mu_0(\invasion)=0}$ and ${\mu_1(\invasion)=1}$, there are two critical values of the parameter $p$ to consider: ${\pcriticalzero = \sup \{p \mid \mu_p(\invasion)=0\}}$ and ${\pcriticalone = \inf\{p \mid \mu_p(\invasion)=1\}}$ which verify ${\pcriticalzero\leq \pcriticalone}$.

\begin{remark}\label{rem:ergodicity}
  In many proofs, an event is shown to be of probability either $0$ or $1$ for a given measure $\mu$ using 0-1 laws.
  
  One of them is the Kolmogorov 0-1-law; but it applies to tail events, which $\invasion$ is not: it might depend on a finite portion of the configuration.
  
  Another of them is that, if the graph $G$ has enough symmetry, Bernoulli measures follow a 0-1 law on Borel sets of configurations that preserve these symmetries.
  
  More precisely, denoting $\Gamma$ the group of automorphisms of graph $G$, if the action of $\Gamma$ on $G$ has an infinite orbit (intuitively meaning that infinitely many vertices of $G$ are indistinguishable) and if $G$ is locally finite, then, for any Bernoulli measure $\mu$ and any set $X$ of configurations which is invariant under $\Gamma$ (\textit{i.e.}, ${c\in X}$ and $\phi\in\Gamma$ implies ${c\circ\phi^{-1}\in X}$), we have ${\mu(X)\in\{0,1\}}$ (see \cite[Lemma 1, chapter 5]{bollob_s2006}).
  
  Whatever the graph $G$, the set $\invasion$ is always invariant under $\Gamma$, simply because ${F(c)\circ\phi^{-1}=F(c\circ\phi^{-1})}$ for any configuration $c$ and $c\in I$ is determined by a uniform constraint on all vertices: ${\forall v\in\sites, c^\infty(v)=1}$.
  Consequently, for this 0-1 law to hold on $\invasion$, the only requirement is for the graph to have enough symmetry. It is the case in particular for Cayley graphs of infinite groups (typically $\ZZ^2$), and is often expressed as the fact that the action of translation with Bernoulli measure is ergodic \cite[Theorem 1.12]{walters1981}.
\end{remark}

Following remark~\ref{rem:ergodicity}, in the case of a graph with enough symmetry, ${\mu_p(I)\in\{0,1\}}$ for all $p$ so that ${\pcriticalzero=\pcriticalone}$.
We will however consider graphs that have few or no symmetry and there is a priori no reason to expect this 0-1 law to hold (see counter-example in appendix \ref{appendix:nonzeroone} and Proposition~\ref{prop:bootstrapnonzeroone} therein).

Furthermore, this 0-1 law uses Bernoulli measures: percolation theory often focuses on these measures to study critical values of parameter $p$.
Here, we will actually consider a larger class of measures in Section~\ref{sec:mainresult} that are not necessarily product measures but are sufficiently well-behaved.

Our goal here is not to look for abstract generality, but rather to make explicit and clear the requirements we use in the probabilistic arguments of our main result (Section~\ref{sec:probargs}).
There are three such requirements that we detail below: bounded correlations (Markov property), non-vanishing probability of patches of 1 of any fixed size, and positive correlation of upward-closed sets. 

Fix some integer ${k\in\NN}$. We say that a measure $\mu$ is \emph{$k$-Markov} if, for any finite ${D\subseteq V}$ and any set ${X\subseteq V}$ which is at distance at least $k$ from $D$ in graph $G$ and any patterns ${u:D\to\{0,1\}}$ and ${P:X\to\{0,1\}}$, the following holds:
\[\mu([u]\cap [P]) = \mu([u])\mu([P]).\]

In addition, we say a measure is \emph{non-vanishing} if for all $r$ there is some constant ${\alpha>0}$ such that for all ${v\in V}$, ${\mu([1^{B(v,r)}])\geq\alpha}$ where ${B(v,r)}$ is the ball of radius $r$ centered in $v$ in graph $G$.
Finally, a set of configurations ${X\subseteq\configurations}$ is \emph{upward-closed} if, whenever $x\in X$  and ${x_v\leq y_v}$ for all ${v\in V}$, then ${y\in X}$.

A measure $\mu$ is called \MNVPC{} (Markov-Non Vanishing-Positively Correlated) if it is $k$-Markov for some $k$, non-vanishing, and any pair $X,Y$ of upward-closed Borel sets are positively correlated, \textit{i.e.} ${\mu(X\cap Y)\geq\mu(X)\mu(Y)}$.

\begin{lemma}
  Any non-trivial Bernoulli measure on a bounded degree graph is \MNVPC{}.
\end{lemma}
\begin{proof}
  Any such measure is clearly $1$-Markov. It is non-vanishing on any graph of bounded degree because the cardinal of balls of radius $r$ is uniformly bounded.
  Finally, Bernoulli measures are positively correlated on upward-closed events, a fundamental fact often referred to as Harris inequality (see \cite[Lemma 3 of chapter 2]{bollob_s2006} for a complete proof).
\end{proof}


One of the basic property of \MNVPC{} measures is that imposing a large patch of $0$ at a fixed position has a decreasing probability with the patch size that can be uniformly exponentially upper-bounded.

\begin{lemma}\label{lem:expdecay}
  If $G$ is of bounded degree and $\mu$ a \MNVPC{} measure, then there is some ${\beta}$ with ${0<\beta<1}$ such that, for any ${D\subseteq V}$ with ${|D|=n}$ it holds that ${\mu([0^D])\leq\beta^n}$.
\end{lemma}
\begin{proof}
  Let $k$ be such that $\mu$ is $k$-Markov, and $\alpha>0$ be such that ${\mu([1^{\{v\}}])\geq\alpha}$ for all $v$ ($\alpha$ exists because $\mu$ is non-vanishing).
  If $\Delta$ is a bound on the degree of $G$ then any ball of radius $k$ in $G$ has cardinality at most ${\Delta^k}$.
  Therefore, one can choose at least ${N=\lfloor\frac{n}{\Delta^k}\rfloor}$ vertices ${v_1,\ldots, v_N}$ in $D$ which are pairwise separated by distance at least $k$.
  By the $k$-Markov property we get
  \[\mu([0^D])\leq\prod_{1\leq i\leq N}\mu([0^{\{v_i\}}])\leq (1-\alpha)^N.\]
  Taking ${\beta = (1-\alpha)^{\frac{1}{2\Delta^k}}}$ proves the lemma.
\end{proof}

\section{Rhombus tilings}
\label{sec:rhombus}
We call \emph{tiling}, denoted by $\tiling$, a countable set of tiles that covers the euclidean plane $\mathbb{R}^2$without overlap. That is, $\tiling = \{t_i, i \in \NN\}$ is a tiling when $\bigcup\limits_{i\in\NN} t_i = \RR^2$ and for any $i,j\in\NN$, $\mathring{t_i}\cap \mathring{t_j} = \emptyset$.
In all that follow, we consider the case of \emph{rhombus tilings}, where there are finitely many tiles up to translation, all the tiles are rhombuses and the tiling is edge-to-edge, \emph{i.e.}, any two tiles either intersect on a single common vertex, on a full common edge, or not at all.
Throughout the article we use the famous example of Penrose rhombus tilings \cite{penrose1979}, see Fig.~\ref{fig:tiling}.

\begin{figure}[htp]
  \includegraphics[width=\textwidth]{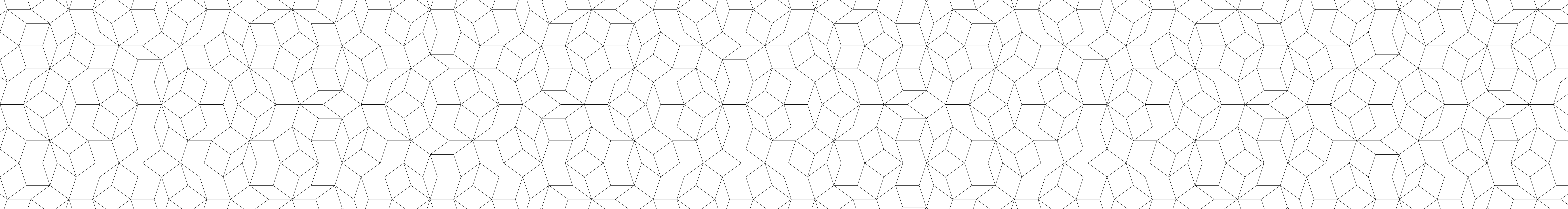}
\caption{A fragment of a Penrose rhombus tiling.} 
\label{fig:tiling}
\end{figure}

In the case of edge-to-edge tilings, the condition that there are finitely many tiles up to translation is equivalent to finite local complexity (FLC) \cite{kenyon1992} that is: for any given size there are finitely many patches of that size.

We say that two tiles $t$ and $t'$ are \emph{adjacent}, denoted by $t\adj t'$, when they share a full edge, that is, when there exists an edge $e$ such that $t \cap t' = e$.
This notion, also called edge-adjacency, is the natural notion of adjacency on edge-to-edge tilings.
However, we also define a weaker notion of vertex-adjacency.
We say that two tiles $t$ and $t'$ are \emph{vertex-adjacent}, denoted by $t\vadj t'$,  when they are distinct and share at least one vertex, \emph{i.e.} $t \cap t' \neq \emptyset$ and $t \neq t'$.

Given a tiling $\tiling$ and a tile $t\in \tiling$, we define the \emph{neighbourhood} of $t$, denoted by $\neighbour{t}$, as the set of tiles which are adjacent to $t$, \emph{i.e.}, $\neighbour{t}:= \{ t' \in \tiling \mid t \adj t'\}$.
We also define the weaker \emph{vertex-neighbourhood} of $t$, $\vneighbour{t}:= \{t' \in \tiling \mid t \vadj t'\}$.
We define similar notions for a set of tiles $S$, that is, $\neighbour{S}:= \{ t' \in \tiling\setminus{S} \mid \exists t \in S, t \adj t'\}$ and $\vneighbour{S}:= \{ t' \in \tiling\setminus{S} \mid \exists t \in S, t \vadj t'\}$. They are respectively the tiles adjacent and vertex-adjacent to the set $S$.
Note that all these notions are not neighbourhoods in the \emph{topological} sense, notably since they do not contain the original set. However, this denomination, which is ultimately about adjacent tiles, is widely used and understood in tilings theory.

We call \emph{patch} $\patch$ a set of adjacent tiles in a tiling. We also sometimes consider \emph{vertex-patches}, which are sets of vertex-adjacent tiles.
We consider only finite patches unless specified otherwise.

The key structure in rhombus tilings are the chains which generalize the rows and columns of the square grid.

\begin{definition}[Chain $\chain$]
  Given a rhombus tiling $\tiling$.

  Given an edge direction $\vec{e}$, we say that two tiles $t$ and $t'$ are $\vec{e}$-adjacent, denoted by $ t \overset{\vec{e}}{\sim} t'$,  when they share an edge of direction $\vec{e}$.

  Given an edge direction $\vec{e}$ and a normal vector $\vec{n}$ (orthogonal to $\vec{e}$), we say that a tile $t'$ is $\vec{e}$-adjacent to $t$ in direction $\vec{n}$ when $t \overset{\vec{e}}{\sim} t'$ and $\langle \vec{tt'} | \vec{n} \rangle > 0$ where $\vec{tt'}$ is the vector from the barycenter of $t$ to the barycenter of $t'$ and $\langle | \rangle$ denotes the scalar product.

  We call \emph{chain} $\chain$ of edge direction $\vec{e}$ a bi-infinite patch $\chain = (t_i)_{i\in \ZZ}$ such that there exists $\vec{n}$ orthogonal to $\vec{e}$ such that for any $i$, $t_{i+1}$ is $\vec{e}$-adjacent to $t_i$ in direction $\vec{n}$.

  We call \emph{half-chain} $\chain^+$ of edge direction $\vec{e}$, orientation $\vec{n}$ and  starting tile $t$ the infinite patch $\chain^+ = (t_i)_{i\in \NN}$ such that $t_0 = t$ and for any $i$, $t_{i+1}$ is $\vec{e}$-adjacent to $t_i$ in direction $\vec{n}$.

  We call \emph{chain segment} a finite connected subset of a chain. For a chain $\chain$ and $i < j$, we write $\chain_i^j$ to denote the subset $\{t_i, \dots, t_j\}$.
\end{definition}

We denote $\chain \equiv \chain'$ when two chains are identical up to shift of index and/or multiplication of one's indices by $-1$.

\begin{figure}[htp]
  \center
  \includegraphics[width=0.8\textwidth]{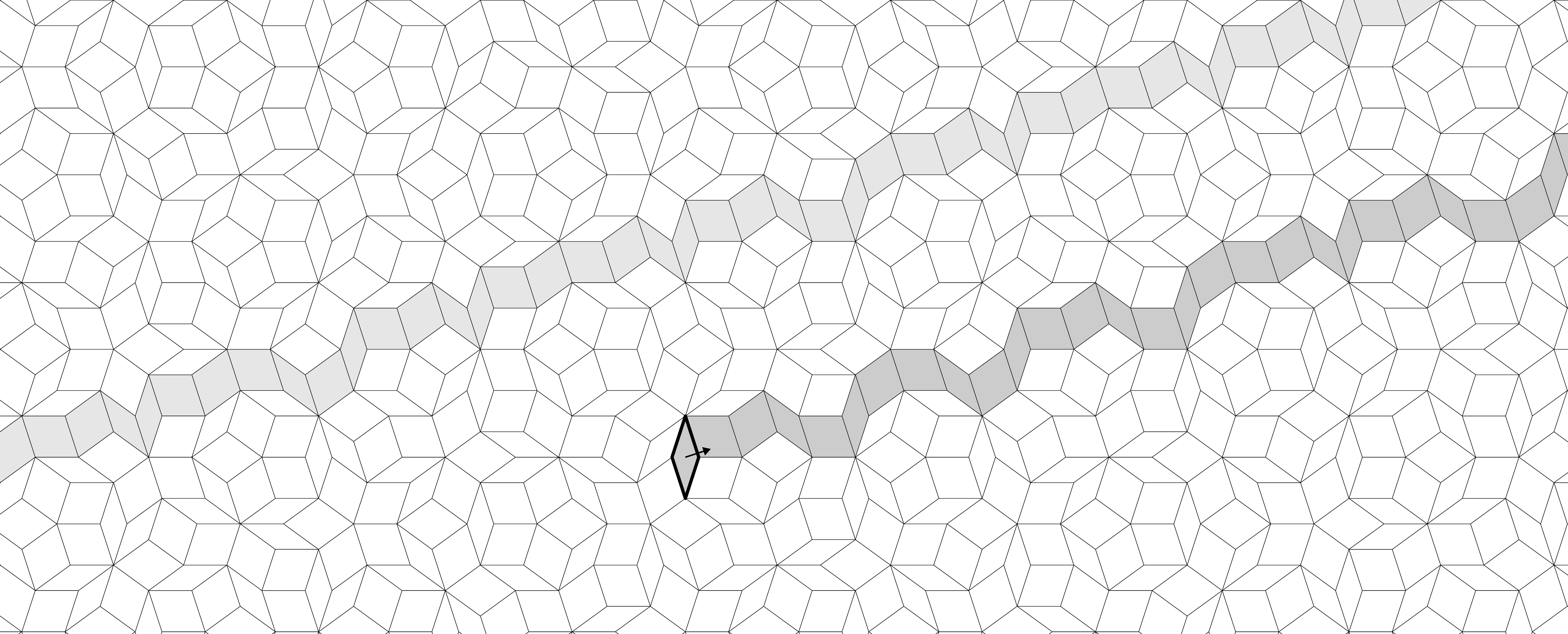}
  \caption{In light grey a chain of rhombuses, in darker gray and with the starting tile in bold a half-chain of rhombuses.}
  \label{fig:chain}
\end{figure}

\begin{lemma}[Chain crossing \cite{kenyon1993}]
  Given a rhombus tiling $\tiling$ and two chains $\chain$ of direction $\vec{e}$ and $\chain'$ of direction $\vec{e'}$.
  The following hold:
  \begin{enumerate}
  \item if $\chain \not\equiv \chain'$ then $\chain$ and $\chain'$ intersect at most once;
  \item if $\chain \not\equiv \chain'$ and $\chain \cap \chain' \neq \emptyset$ then the intersection tile $t$ has edge directions $\vec{e}$ and $\vec{e'}$
  \item if $\vec{e} = \pm \vec{e'}$ then either $\chain \equiv \chain'$ or $\chain \cap \chain' = \emptyset$
  \end{enumerate}
\end{lemma}

Note however that, in an arbitrary rhombus tiling, non-intersecting chains might have different edge directions. 

\begin{lemma}[Uniform monotonicity]
	\label{lemma:uniform_monotonicity}
  Let $\tiling$ be a rhombus tiling with finitely many edge directions.
  There exists a constant $\theta>0$ such that any chain $\chain$ is $\theta$-uniformly monotonous, that is, if $\chain=(\chain_i)_{i\in \ZZ}$ has direction $\vec{e}$, there exists a unit normal vector $\vec{n}$ orthogonal to $\vec{e}$ such that for any $i\in\ZZ$, $\langle \vec{t_{i}t_{i+1}} | \vec{n} \rangle \geq \theta$ where $\vec{t_{i}t_{i+1}}$ is the vector from the barycenter of $t_{i}$ to the barycenter of $t_{i+1}$.
\end{lemma}

\begin{proof}
  Let $\tiling$ be a rhombus tiling with $d$ edge directions $\{ \vec{e_0}, \dots \vec{e_{d-1}}\}$.
  Denote $\vec{e}^\bot$ the unit vector orthogonal to $\vec{e}$ in the positive orientation.

  Let $\theta := \min\limits_{0\leq i,j < d, i\neq j} |\langle \vec{e_i} | \vec{e_j}^\bot \rangle|$.

  Any chain $\chain \subset \tiling$ is $\theta$-uniformly monotonous.
  Indeed, let $\chain$ be a chain of edge direction $\vec{e}$.
  Let $t$ and $t'$ be two consecutive tiles in $\chain$.
  There exists two directions $\vec{e_i}$ and $\vec{e_j}$ (possibly the same) such that $t$ has edge directions $\vec{e}$ and $\vec{e_i}$ and $t'$ has edge directions $\vec{e_j}$ and $\vec{e}$.
  Without loss of generality (changing the choice of orientations of the directions), assume $\langle\vec{e_i}|\vec{e}^\bot\rangle >0$ and $\langle \vec{e_j} | \vec{e}^\bot \rangle > 0$.
  Since $t$ and $t'$ are parallelograms and adjacent along an $\vec{e}$ edge, this means that $\vec{tt'} = \tfrac{1}{2}(\vec{e_i} + \vec{e_j})$.
  So $\langle \vec{tt'} | \vec{e}^\bot \rangle = \tfrac{1}{2}(\langle\vec{e_i}|\vec{e}^\bot\rangle  + \langle\vec{e_j}|\vec{e}^\bot\rangle )\geq \theta$.   

\end{proof}
\begin{figure}[htp]
  \center
  \includegraphics[width=0.8\textwidth]{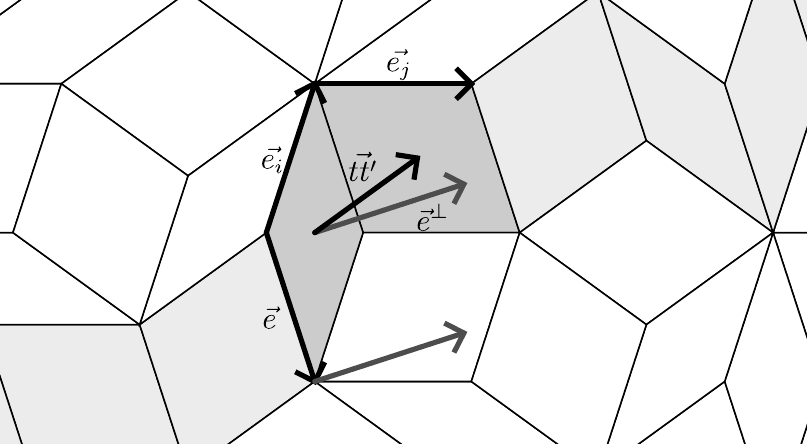}
  \caption{Adjacent tiles along a chain and vector.}
  \label{fig:adj_tiles}
\end{figure}

\begin{remark}[Rhombus and parallelograms]
  The keen reader might have already noticed that the combinatorial structure of chains only uses the fact that rhombuses have two pairs of opposite parallel edges, that is, they are parallelograms.
  However for simplicity of redaction we chose to only mention rhombus tilings as any edge-to-edge parallelogram tiling is combinatorially equivalent to an edge-to-edge rhombus tiling, meaning they have the same adjacencies.
  Indeed the transformation consisting in rescaling all edge directions so they have the same length transforms an edge-to-edge parallelogram tiling to a combinatorially equivalent edge-to-edge rhombus tiling.
  Note however that, by transforming a tileset in this manner we might allow more tilings as collinear edge directions merge in this process.
\end{remark}

\section{Critical probability for 2-bootstrap percolation on rhombus tilings}
\label{sec:mainresult}

The 1-bootstrap dynamics is not interesting on connected graphs ($\invasion$ contains all configurations except one). Moreover, on rhombus tilings, the 3-bootstrap dynamics always admits finite obstacles and therefore has a trivial critical probability (${p_c^0=p_c^1=1}$): indeed, considering any vertex of the tiling and assuming that all tiles sharing this vertex are in state $0$ (not infected), then they will stay in state $0$ forever, because each of them has at most $2$ neighbours in state $1$ at any moment.

In this section we investigate 2-neighbour contamination on rhombus tilings, and we prove almost sure invasion on any rhombus tiling with finitely many edge directions for any \MNVPC{} measure $\mu$.
In particular, it implies almost sure invasion for any non-trivial Bernoulli measure so that the critical percolation threshold of 2-neighbour contamination on any rhombus tiling with finitely many edge directions  is $p_c^0=p_c^1=0$.

Formally, to fit the formalism of Section \ref{sec:bootstrap-intro}, we can consider the adjacency graph of the rhombus tilings.
The adjacency graph, or dual graph, $G_\tiling$ of the rhombus tiling $\tiling$ is $G_\tiling = (\tiling, E_\tiling)$ where $(t,t')\in E_\tiling \Leftrightarrow t\adj t'$.
However, for simplicity, we identify $\tiling$ and $G_\tiling$ in this section.

In Section \ref{sec:geometrical_elements}, we study the geometry and combinatorics of \emph{clusters} (stable contaminated patterns) and prove that any finite cluster is enclosed in a ``wall'' which is a polygon of a bounded number of chain segments, generalising the fact that in $\mathbb{Z}^2$ clusters are enclosed in a rectangle that is in particular a polygon of 4 chain segments.
We also prove that infinite clusters are bounded by some chain path including an infinite half-chain.

In Section \ref{sec:probargs} we use the geometrical and combinatorial results on clusters to prove that invasion is almost sure for any \MNVPC{} measure.
The key idea is counting the possible finite ``walls'' of length $n$ around a tile and bounding their number by a polynomial in $n$.
From this we deduce for any \MNVPC{} measure $\mu$ :
\begin{itemize}
  \item a uniform bound $\lambda<1$ on the measure of the event $E_t$ of a tile being enclosed in a finite ``wall'', and
  \item a uniform approximation of $E_t$ by the tail event $E_{t,\geq n}$ of the tile $t$ being enclosed in a finite ``wall'' of length more than $n$.
\end{itemize}
Remarking that the head events $E_{t,\leq n}$ of being enclosed in a wall of length less than $n$ are independent for sufficiently far away tiles for Markov measures, we obtain almost independence of far away events $E_t$ from which we conclude almost sure invasion for any \MNVPC{} measure.

\subsection{Geometrical and combinatorial elements}
\label{sec:geometrical_elements}
The first observation that we can make is that a chain of 0s can stop a contaminated cluster, see Fig~\ref{fig:chain-wall}. Actually we prove that chains of $0$s are the only possible ``walls'' for $2$-neighbour contamination.
To formalize this we introduce the notion of chain-convexity.

\begin{figure}[htp]
  \center
  \includegraphics[width=0.48\textwidth]{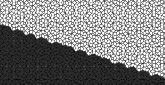}
  \hfill
  \includegraphics[width=0.48\textwidth]{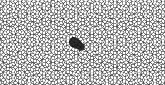}
  \caption{On the left a stable half plane of $1$s (represented in dark gray), on the right a simple stable finite cluster of $1$s.}
  \label{fig:chain-wall}
\end{figure}

\begin{definition}[Chain convex]
  A connected set of tiles $S\subset \tiling$ is called \emph{chain-convex} when for any chain of tiles $\chain=(\chain_i)_{i\in \ZZ}$ the following hold:
  \begin{enumerate}
  \item $\{j \in \ZZ \mid \chain_j \in S\}$ is an interval of $\ZZ$;
  \item if there exist $i< j$ such that $\chain_i$ and $\chain_j$ are vertex-neighbours of $S$ (that is $\{\chain_i, \chain_j\} \subset \vneighbour{S}$) then either $\chain_{i+1}^{j-1} \subset \neighbour{S}$ and no tile of $\chain$ is in $S$, or $\neighbour{S}\cap \chain = \{ \chain_i, \chain_j\}$ and the tiles of $\chain$ that belong to $S$ are exactly $\chain_{i+1}^{j-1}$.
  \end{enumerate}
\end{definition}

Note that a chain convex set of tiles, as it is connected and contains no hole due to item 1 of the definition, is simply connected.

\begin{remark}
	If a tile $\chain_i$ from a chain $\chain$ is vertex-adjacent to a set $S$, then among $\chain_{i+1}$ and $\chain_{i-1}$, one of them is either vertex-adjacent to $S$ too, or in $S$. Indeed, $\chain_{i+1}$ shares half the vertices of $\chain_i$ with it, and $\chain_{i-1}$ the other half; consequently one of them also shares a vertex with a tile of $S$, unless it is itself a tile of $S$.
\end{remark}

\begin{remark}[Intersections of chains with the boundary of chain-convex sets]
  This definition states that if $\chain$ is a chain and $S$ is chain-convex, we have the following possibilities for the adjacency of $\chain$ and $S$:
  \begin{itemize}
  \item $\chain$ has no element that is vertex-adjacent to $S$ (this also holds if $\chain \subset S$);
  \item $\chain$ has exactly one element $\chi_{i}$ that is vertex-adjacent to $S$; then by the above remark either $\{ \chi_{i+1}, \chi_{i+2}, \dots \}$ or $\{ \chi_{i-1}, \chi_{i-2}, \dots \}$ is in $S$;
  \item $\chain$ has exactly two elements $\chain_i$ and $\chain_j$ that are vertex-adjacent to $S$ (and in particular no other element can be adjacent to $S$); then $\chain_i$ and $\chain_j$ are adjacent to $S$ and the intersection with $S$ is exactly the (possibly empty) set $\{ \chain_{i+1}, \dots, \chain_{j-1} \}$;
  \item $\chain$ has at least three elements that are vertex-adjacent to $S$ (and notably not in $S$), and a finite number of them; then let us take $\chain_i$ and $\chain_j$ be the ones with smallest and largest indices. We cannot have the second case of point 2 of the definition above, therefore $\chain$ does not intersect $S$, all elements in $\{ \chain_{i+1}, \dots, \chain_{j-1} \}$ neighbour $S$, and no element of $\chain$ outside of $\chain_i^j$ is adjacent or vertex-adjacent to $S$.
  \item $\chain$ has infinitely many elements that are vertex-adjacent to $S$. Then similarly $\chain$ does not intersect $S$, and either has a half-chain adjacent to $S$, or is adjacent to $S$ in its entirety.
  \end{itemize}
\end{remark}

With the following two results, we prove that any vertex-connected set of tiles that is stable for $2$-neighbour contamination is chain convex, as we would want it to be. We start by a more technical lemma on vertex-connected stable sets of tiles.

\begin{lemma}[Technical result on chains and stable clusters]\label{lemma:technical_stablechain}
  Let $S$ be a  vertex-connected set of tiles $S\subset \tiling$ that is stable for $2$-neighbour contamination.
  Let $\chain \in \tiling$ be a chain of rhombuses of edge direction $\edge$ such that there exist $i<j \in \ZZ$ with $\chain_i, \chain_j \in \vneighbour{S}$.
  Denote $v_i$ (resp. $v_j$) a vertex common to $\chain_i$ and $S$ (resp. common to $\chain_j$ and $S$).

  If no tile of $\chain_{i+1}^{j-1}$ is in $S$, and if $v_i$ and $v_j$ are in the same half-space of $\tiling \setminus \chain$,
  (that is, $v_i$ and $v_j$ are connected in $S\cap (\tiling \setminus \chain)$, as in Fig.~\ref{fig:technical_stablechain}),
  then the edge path $\epath_{\chain}$ from $v_i$ to $v_j$ along the boundary of $\chain$ is also in the boundary of $S$.
\end{lemma}

Note that in the following proofs, $\tiling \setminus \chain$ and $S \cap (\tiling \setminus \chain)$ are understood as sets of tiles.

\begin{proof}
  Let $S$ be a vertex-connected set of tiles $S\subset \tiling$ that is stable for $2$-neighbour contamination.
  Let $i<j \in \ZZ$ and $\chain=(\chain_i)_{i\in \ZZ}$ a chain of rhombuses of edge direction $\edge$ satisfying the condition above.
  Notably, $\chain_{i+1}^{j-1}$ has no tile in $S$; and there are vertices $v_i$ and $v_j$ vertices of $\chain_i\cap S$ and $\chain_j \cap S$ which are in the same half-space of $\tiling \setminus \chain$, and connected in $S\cap (\tiling \setminus \chain)$ as depicted in Fig.~\ref{fig:technical_stablechain}.
  
  We prove that all the rhombuses $\chain_k\in\chain_{i+1}^{j-1}$ are adjacent to $S$ through a non-$\edge$ edge.
  More precisely, the edge path $\epath_\chain$ from $v_i$ to $v_j$ along the boundary of $\chain$ (it exists as $v_i$ and $v_j$ are in the same half-space) is also in the boundary of $S$.

  \begin{figure}[htp]
    \includegraphics[width=\textwidth]{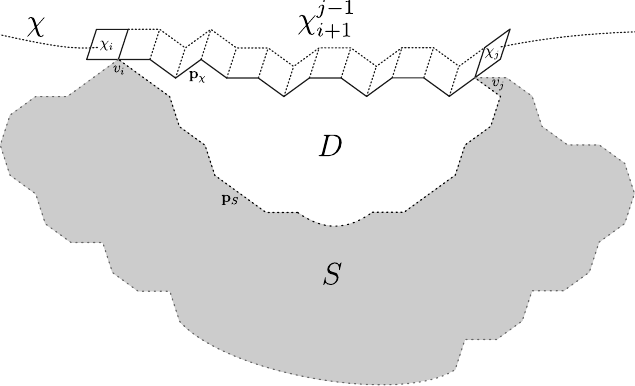}
    \caption{A chain $\chain$ touching twice a stable set $S$. We prove that the domain $D$ must be empty.}
    \label{fig:technical_stablechain}
  \end{figure}
  
  Denote $\epath_S$ a path from $v_j$ to $v_i$ along the boundary of $S$ such that $S$ does not intersect the domain delimited by $\epath_\chain\cdot \epath_S$, where $\cdot$ is the path concatenation. To obtain such a path, one can first take an arbitrary path $p_0$ along the boundary of $S$ connecting $v_j$ to $v_i$ ($v_j$ and $v_i$ lie on the boundary of $S$ by hypothesis, since $\chain_i, \chain_j \in \vneighbour{S}$). If the (finite) domain $D_0$ delimited by $\epath_\chain\cdot p_0$ doesn't intersect $S$ we are done, otherwise there must be a tile $t\in S\cap D_0$ with a vertex on $p_0$. Considering the edge-connected component of $S\cap D_0$ containing $t$, we can define a new path $p_1$ along the boundary of $S$ still connecting $v_j$ to $v_i$ but such that the domain $D_1$ delimited by $\epath_\chain\cdot p_1$ is included in $D_0$ but does not contain $t$.
  
  Iterating this process a finite number of times (since $S \cap D_0$ is finite), we obtain the desired path $\epath_S$.

  First, remark that if $\epath_S$ is trivial, that is if $v_i=v_j$, then by uniform monotonicity of $\chain$ in direction $\edge^\bot$ (Lemma \ref{lemma:uniform_monotonicity}), $j=i+1$ and so the conclusion holds vacuously with $\chain_{i+1}^{j-1} = \emptyset$.

  Now we assume that $\epath_S$ is non trivial, as depicted in Fig.~\ref{lemma:technical_stablechain}.
  Assume for contradiction that the finite patch of tiles $D$ defined as the interior of the cycle $\epath= \epath_\chain\cdot \epath_S$ is non empty.
  We now construct a subset $D'\subsetneq D$ that satisfies the same hypothesis and has to be non-empty. We consequently reach a contradiction by iterating this construction which builds a strictly decreasing sequence of non-empty finite sets of tiles.

  As $D$ is non-empty and delimited by $\epath_S$, there exists a tile $t\in D\cap \neighbour{S}$. As $S$ is stable for $2$-neighbour contamination, $S$ touches exactly one edge of $t$. So $t$ has another edge direction $\vec{e_t}$ so both $\vec{e_t}$-edges of $t$ are not in $\epath_S$. Denote $\chain'$ the chain of rhombuses of edge direction $e_t$ passing through $\chain'_0:=t$.
  As $\epath$ is a closed curve containing $t$, and $\chain'$ is uniformly monotonous in direction $e_t^\bot$, $\chain'$ crosses $\epath$ at least twice. As chains can cross at most once, the $\epath_{\chain}$ part of $\epath$ cannot be crossed twice by $\chain'$, and so $\chain'$ crosses $\epath_S$ at least once.
  Therefore there exists $t':=\chain'_k$ such that ${\chain'}_0^{k}\subset D$ and $\chain'_{k+1}\in S$.
  In particular, $\chain'_k$ is adjacent to $S$ through an $e_t$ edge.
  Denote $D'$ the subset of $D$ delimited by the subpath $\epath'_S$ between $t$ and $t'$ along the boundary of $S$ and the path $\epath_{\chain'}$ between $t$ and $t'$ on the boundary of $\chain'$, as depicted in Fig.~\ref{fig:technical_stablechain_decomposition}.
  $D'\subsetneq D$ and $D'\neq \emptyset$ as otherwise $t'$ would be adjacent to $S$ through both a $e_t$ edge and a non-$e_t$ edge, contradicting the stability of $S$.
  Additionnally, the chain $\chain'$ from $t$ to $t'$ satisfies the same hypothesis as the chain $\chain$ from $\chain_i$ to $\chain_j$. 
  So we can repeat this decomposition process to reach a contradiction.

  \begin{figure}[htp]
    \includegraphics[width=\textwidth]{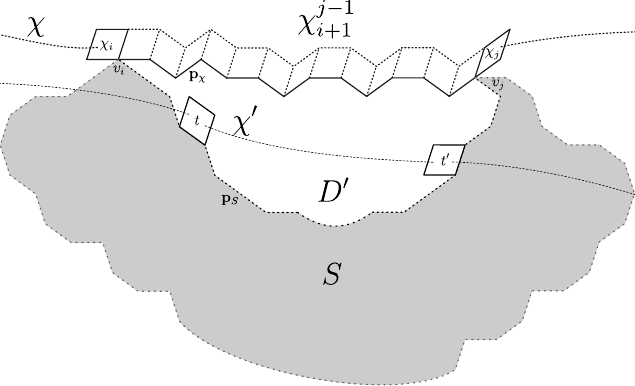}
    \caption{Decomposing the domain $D$.}
    \label{fig:technical_stablechain_decomposition}
  \end{figure}

  Therefore $\epath_S=\epath_\chain$ and all the tiles along the chain $\chain$ from $v_i$ to $v_j$ are adjacent to $S$ through a non-$\edge$ edge.
\end{proof}

\begin{remark}
A direct consequence of that lemma is that if a chain $\chain$ (of direction $\edge$) touches a connected set $S$ of tiles that is stable for $2$-neigbour contamination through two vertices $v_i \in \chain_i$ and $v_j \in \chain_j$ lying in two different half-spaces of $\tiling \setminus \chain$, then the whole portion of chain $\chain_{i+1}^{j-1}$ is in $S$.

This follows from the fact that $v_i$ and $v_j$ are connected in $S$ by some edge path $\epath$, which crosses $\chain$ by Jordan's theorem.

Denote $v_{i_1},\dots, v_{i_k}$ the vertices of intersection of $\epath$ with $\chain$, including $v_i$ and $v_j$, named from the smallest index $i_1$ to the biggest index $i_k$ of tiles in the chain $\chain$ that contain these vertices (up to using $\prime$ if a given tile has several vertices involved in the process). Applying Lemma~\ref{lemma:technical_stablechain} on the elementary intervals between consecutive intersection points, we obtain that all tiles of $\chain_{i_1+1}^{i_k-1}$ are adjacent to $S$ through a non-$\edge$ edge.

Additionally, if $\epath$ crosses $\chain$, then $S$ also contains some tile $\chain_{k'} \in \chain$, with $i_1 \leq k' \leq i_k$. Then $\chain_{k'+1}$ (or $\chain_{k'-1}$ if $k' = i_k$) is adjacent to $S$ both through $\chain_{k'}$ and a non-$\edge$ edge. Hence $\chain_{k'+1} \in S$ (or $\chain_{k'-1}$) by stability of $S$ for 2-neighbour contamination. With a finite number of steps, we deduce that the entire $\chain_{i_1+1}^{i_k-1}$ chain portion is in $S$. Notably, $\chain_{i+1}^{j-1}$ is in $S$.

\label{remark:technical_corollary}
\end{remark}

\begin{lemma}[Bootstrap stability]\label{lemma:stablechainconvex}
  Any vertex-connected set of tiles $S\subset \tiling$ that is stable for $2$-neighbour contamination is chain convex.
\end{lemma}

\begin{proof}
  Let $S\subset \tiling$ be a vertex-connected set of tiles that is stable for $2$-neighbour contamination.

  We prove, using Lemma \ref{lemma:technical_stablechain}, that $S$ is chain-convex, by checking all the items of that definition.
  Let $\chain$ be a chain of rhombuses in $\tiling$.

  We prove that $\{j \in \ZZ, \chain_j \in S\}$ is an interval of $\ZZ$.
  Let $\chain_j\in S$ and $\chain_{j+k}\in S$.
  If $\chain_j$ and $\chain_{j+k}$ are not connected in $S\cap(\tiling \setminus \chain)$, then as remarked in Remark \ref{remark:technical_corollary}, we have $\chain_{j+1}^{j+k-1} \subset S$.
  Otherwise, denote $\epath$ the edge path from a vertex $v_j$ of $\chain_j$ to a vertex $v_{j+k}$ of $\chain_{j+k}$ in the same half-space such that $v_j$ and $v_{j+k}$ are connected in $S\cap (\tiling \setminus \chain)$.
  For contradiction, assume that $\chain_{j+1}^{j+k-1}$ is not entirely in $S$. This means it contains a sub-interval that is entirely not in $S$. Up to renaming, we keep the notations $j$ and $j+k$ for that sub-interval, taken as big as possible (so that we still have $\chain_j$ and $\chain_{j+k}$ in $S$).
  By Lemma \ref{lemma:technical_stablechain}, if $\chain_{j+1}^{j+k-1}$ does not intersect $S$ then $\epath$ is in the boundary of $S$. Moreover, $\chain_{j+1}\notin S$. And yet, $\chain_{j+1}$ has two edge neighbours in $S$: its neighbour on the other side of the edge path $\epath$, and $\chain_j$. This contradicts the stability of $S$ for $2$-neighbour contamination.
  Consequently, $\chain_{j+1}^{j+k-1}$ is entirely in $S$ and so $\{j \in \ZZ, \chain_j \in S\}$ is an interval of $\ZZ$.
  \smallskip

  We now prove the second chain-convexity condition.
  Let $i<j$ such that $\chain_i$ and $\chain_j$ are in $\vneighbour{S}$ (not in $S$).
  Denote $v_i$ (resp. $v_j$) the vertex of $\chain_i$ (resp. $\chain_j$) the vertex touching $S$.
  As explained in Remark \ref{remark:technical_corollary}, if $v_i$ and $v_j$ are not in the same half-space or not connected in $S\cap (\tiling \setminus \chain)$ then $\chain_{i+1}^{j-1}$ is included in $S$.
  We now assume that $v_i$ and $v_j$ are in the same half-space and connected in $S \cap(\tiling \setminus \chain)$.
  Denote $\epath$ the edge path from $v_i$ to $v_j$ along the boundary of $\chain_i^j$.
  Up to decomposing $\chain_i^j$ in smaller intervals whose endpoints are not in $S$, we assume that $\chain_{i+1}^{j-1}$ either is entirely in $S$ or has no tile in common with $S$, while $\chain_i, \chain_j \notin S$.
  
  In the first case we have $\chain_{i+1}^{j-1}\subset \chain \cap S$ and $\chain_i, \chain_j \notin S$. By the fact we proved above that $\{j \in \ZZ, \chain_j \in S\}$ is an interval of $\ZZ$, we obtain that $S\cap \chain = \chain_{i+1}^{j-1}$. Additionally, the only elements of $\chain$ in $\vneighbour{S}$ are $\chain_i$ and $\chain_j$. Indeed, if it weren't the case, then there exists $k$ (assume by symmetry $k>j$) such that $\chain_k\in \vneighbour{S}$, and then by applying Lemma \ref{lemma:technical_stablechain} we obtain that the edge path $\epath$ from $\chain_j$ to $\chain_k$ is on the boundary of $S$. So $\chain_j$ has two edge neighbours in $S$: its neighbour on the other side of the edge path $\epath$, and $\chain_{j-1}$. This would contradict the stability of $S$.
  
  In the second case, $\chain_{i+1}^{j-1}$ has no tile in common with $S$. By applying Lemma \ref{lemma:technical_stablechain}, we obtain that the edge path $\epath$ from $v_i$ to $v_j$ is on the boundary of $S$ and so $\chain_{i+1}^{j-1}\subset \neighbour{S}$. In that case we also have that no tile of $\chain$ is in $S$. Indeed, if there exist $\chain_k\in S$ (assume by symmetry that $k>j$) then taking the smallest such $k$, we have $\chain_{k-1}\in\vneighbour{S}$, and the edge path $\epath'$ from $v_j$ to $v_k$ (a common vertex to $\chain_k$ and $\chain_{k-1}$) is in the boundary of $S$. And so once again $\chain_{k-1}$ has two neighbours in $S$ ($\chain_k$ and some $t\in S$ through $\epath'$), contradicting stability.

  Overall, we have proved that $S$ is chain-convex.
\end{proof}

We call \emph{fortress} (or finite obstacle) a non-empty finite patch $R$ that resists outside contamination.
That is, the configuration $c_R$ where tiles in $R$ are $0$s and all other tiles are $1$s is stable for $2$-neighbour contamination. See Fig.~\ref{fig:fortress} for an example of fortress for $2$-neighbour contamination with quadrilater tiles.
\begin{figure}[htp]
	\center \includegraphics[width=0.4\textwidth]{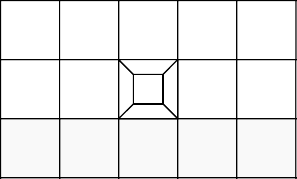}
	\caption{The central pattern of 4 trapezes and the small square forming a big square is a fortress for two neighbour contamination. Indeed, each trapeze has only one outside neighbour, so it cannot be contaminated by the outside if the pattern itself is initially non-contaminated.}
	\label{fig:fortress}
\end{figure}

Straightforwardly, the existence of fortresses in a structure is incompatible with a critical percolation threshold less than $1$.
Our first step to prove the existence of a non-trivial critical percolation threshold is to use Lemma~\ref{lemma:technical_stablechain} to prove the absence of such fortresses.

\begin{lemma}[Fortress]\label{lemma:no-fortress}
	In an edge-to-edge rhombus tiling there is no fortress (finite obstacle) for $2$-neighbour contamination.
\end{lemma}

\begin{proof}
	Let $R$ be a fortress and $S$ be its complement.
	By definition, $S$ is stable for $2$-neigbour contamination.
	
	Let $\chain$ be a chain that intersects $R$. As $R$ is finite, $\chain$ enters and leaves $R$.
	That is, there exist $i<j$ such that $\chain_i,\chain_j \in S$, $\chain_{i+1}, \chain_{j-1} \in R$.
	Up to restricting to a sub-segment of chain, we assume that $\chain_{i+1}^{j-1}\subset R$.
	
	Denoting $v_i$ and $v_i'$ (resp. $v_j$ and $v_j'$) the vertices of $\chain_i \cap \chain_{i+1}$ (resp. $\chain_{j-1}\cap \chain_j$) such that $v_i$ and $v_j$ are in the same half-space delimited by $\chain$.
	Denoting $\epath$ (resp. $\epath'$) the edge path from $v_i$ to $v_j$ along $\chain$ (resp. from $v_i'$ to $v_j'$) and applying Lemma \ref{lemma:technical_stablechain} on both we obtain that the tiles of $\chain_{i+1}^{j-1}$ are edge adjacent to $S$ through both $\epath$ and $\epath'$, contradicting the stability of $S$.
\end{proof}

In what follows, we define some nice notions of polygons and paths in the rhombus tiling, then use Lemma~\ref{lemma:no-fortress} and the results above to prove that any chain-convex patch has such a well-behaved boundary.

\begin{definition}[Chain polygon and polypath] 
  We say that $\polygon=(\polygon_i)_{0\leq i <m}$ is a \emph{chain polygon} in $\tiling$ when :
  \begin{itemize}
    \item the $\polygon_i$ are distinct tiles, 
    \item consecutive $\polygon_i$ are edge adjacent, that is : for any $i$, $\polygon_{i+1\mod m} \in \neighbour{\polygon_i}$,
    \item the cycle has no cords, that is: for any $i,j$ such that $j\neq i \pm 1 \mod m$, $\polygon_j\notin\neighbour{\polygon_i}$.
  \end{itemize}

  Note that any chain polygon $\polygon=(\polygon_i)_{0\leq i < m}$ can be partitoned into $k$ chain segments for some integer $k$, that is: there exists chains $\chain_1,\dots \chain_k$ and indices $i_1,\dots i_{k-1}$ such that $\{\polygon_0, \dots \polygon_{i_1}\} \subset \chain_1$, \linebreak $\{\polygon_{i_1},\dots \polygon_{i_2}\} \subset \chain_2$, \dots $\{\polygon_{i_{k-1}},\dots \polygon_{m-1},\polygon_0\} \subset \chain_k$. Indeed, one can always choose $m$ chain segments made of only two tiles each.

  A chain polygon is called a chain $n$-gon if it can be partitioned into at most $n$ chain segments.

  We say that $\polygon = (\polygon_i)_{i\in \mathbb{Z}}$ is a simple path in $\tiling$ when the $\polygon_i$ are distinct tiles, consecutive $\polygon_i$ are edge-adjacent and there are no cords.

  We call \emph{chain polypath} a simple path $(\polygon_i)_{i\in\ZZ}$ that can be partitioned into $k$ half-chains and chain segments for some integer $k$.
  A chain polypath is called a $n$-polypath if it can be partitioned into at most $2$ half-chains and at most $n-2$ chain segments.
\end{definition}

\begin{lemma}[Chain-convex patches and chain $2d$-gons]
  Let $S\subsetneq \tiling $ be a chain-convex set of tiles, and $d$ be the number of edge directions in the tiling.

  If $S$ is finite, then the exterior tile boundary of $S$ is a chain $2d$-gon.

  If $S$ is infinite, then the exterior tile boundary of $S$ contains an infinite chain $2d$-polypath. More precisely, the exterior tile boundary of $S$ is the disjoint union of one or more chain polypaths of which the sum of the number of chain segments is at most $2d$.  
\end{lemma}

\begin{proof}
  Let $S$ be a finite chain-convex set of tiles.

  First note that we can partition $\vneighbour{S}$ (which is finite, as $S$ is) as a chain-polygon of chains that do not intersect $S$. Indeed, any two consecutive tiles in $\vneighbour{S}$ (when going around $S$ in a given rotationnal order) are edge-adjacent and the chain connecting them does not intersect $S$ (by definition of chain-convexity). So grouping the tiles of $\vneighbour{S}$ by common edge direction, we get finitely many segments of chain.

  Actually, we have at most $2d$ chain segments in $\vneighbour{S}$ because for each edge orientation $\edge$ there are at most two chain segments of edge direction $\edge$ in $\vneighbour{S}$.
  This follows from the uniform monotonicity of chains, see Lemma~\ref{lemma:uniform_monotonicity}, and the fact that chains of same edge direction cannot cross.
  Indeed assume that for some edge direction $\edge$, there are three distinct chains $\chain$, $\chain'$ and $\chain''$ that do not intersect $S$ but have tiles in the vertex-neighbourhood of $S$. Up to renaming assume that $\chain'$ is between $\chain$ and $\chain''$ as in Fig.~\ref{fig:threechains}. By uniform monotonicity, $\chain'$ either crosses $S$ (which is a contradiction) or $\chain$ or $\chain''$ (which is also a contradiction).

  \begin{figure}[htp]
    \center
    \includegraphics[width=0.8\textwidth]{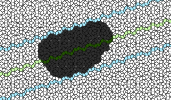}
    \caption{A chain convex set $S$ (in dark grey and dark green) with two adjacent chains of same edge direction (in light blue), any chain of the same direction in between (in green) has to intersect $S$.}
    \label{fig:threechains}
  \end{figure}

  \bigskip

  Let $S$ be an infinite chain-convex set of tiles. Let $P$ be a connected component of $\vneighbour{S}$, we prove that $P$ is an infinite chain $2d$-polypath.
  This proof is decomposed in two parts: first we prove that $P$ is infinite, and second we prove that it intersects non-trivially (meaning for at least two consecutive tiles) at most two chains for each chain direction.
  As $P$ is a connected component of $\vneighbour{S}$ and $S$ is infinite and chain-convex, each tile $t\in P$ has two edge-neighbours in $P$, either two opposite edges (the tiles are along a chain) if $P$ shares an edge with $S$, or two adjacent edges if $P$ only shares a vertex with $S$.
  This implies that if $P$ is finite, then it is a chain polygon. This polygon induces a finite interior and an infinite exterior. As $S$ is infinite it cannot be the interior, and hence $P$ together with its interior forms a fortress (finite obstacle) for $2$-neighbour contamination which contradicts Lemma \ref{lemma:no-fortress}. So $P$ is infinite.
  
  As detailed in the first part of the proof, each chain that has two consecutive tiles in $\vneighbour{S}$ does not insersect $S$.
  This implies, using the same proof as above, that for each chain directions there are at most two segments (or half-chains) of chain of that direction in $P$.
  This implies that $P$ is an infinite chain $2d$-polypath (recall that $2d$-polypath means it has at most $2d$ components).
\end{proof}

Note that in the case of infinite $S$, $S$ might be an infinite strip between two non-intersecting chains of rhombuses. In that case, $\vneighbour{S}$ contains exactly 2 connected components which are both an infinite chain.

\begin{corollary}[Finite stable patches]\label{corollary:bootstrap_stability}
  Any finite patch of tiles $P\subset \tiling$ that is stable for $2$-neighbour contamination is delimited by a chain-polygon with at most $2d$ sides with all $0$s.

  The exterior boundary of any infinite simply-connected set of tiles $P\subsetneq \tiling$ that is stable for $2$-neighbour contamination contains an infinite half-chain with all $0$s.
\end{corollary}

\subsection{Probabilistic arguments}
\label{sec:probargs}

In this subsection, we fix a rhombus tiling $T$ and a \MNVPC{} measure $\mu$.

Let $\chainzero$ be the set of configurations that contain at least a half-chain of zeros,\emph{i.e.},
\[ \chainzero := \{ x \in \configurations \mid \exists \text{ a half-chain } \chain=(\chain_i)_{i\in\NN}, \forall i\geq 0, x_{\chain_i}=0\}\].

\begin{lemma}
	\label{lemma:muchainzero}
  $\mu(\chainzero) = 0$.
\end{lemma}

\begin{proof}
  For a fixed half-chain $\chain=(\chain_i)_{i\in\NN}$ the set ${\chainzero_\chain := \{ x \in \configurations \mid \forall i\geq 0, x_{\chain_i}=0\}}$ has measure $0$ because
  \[\chainzero_\chain = \bigcap_{i\in\NN} [0^{\chi_i}]\]
  and ${\displaystyle\mu\bigl(\bigcap_{0\leq i\leq n}[0^{\chi_i}]\bigr)\leq\beta^n}$ by Lemma~\ref{lem:expdecay} where ${0<\beta<1}$.

  Since there are only countably many half-chains (a half-chain is uniquely identified by an initial tile, a direction and an orientation), the union bound gives ${\mu(\chainzero)=0}$ because
  \[\chainzero = \bigcup_{\chain\text{ half-chain}}\chainzero_\chain.\]
\end{proof}

Let $\blocking$ be the set of configurations that do not invade but do not contain any half-chain of zeros, \emph{i.e.},
\[ \blocking := \invasion^\complement \cap \chainzero^\complement\]

Recall that we call chain $2d$-gon a chain polygon with at most $2d$ sides.

\begin{lemma}
	\label{lemma:blockingenclosed}
  In the finitely blocking configurations $\blocking$, every $1$ is enclosed in a closed chain $2d$-gon of $0$, \emph{i.e.},
  \[ \blocking \subset \{ x \in \configurations \mid \forall v \in \tiling, [x_v=1 \Rightarrow \exists \text{closed chain $2d$-gon } c \text{ around } v  \text{ of } 0]\}\]
\end{lemma}

\begin{proof}
  This derives from the definition of $\blocking$ together with Corollary~\ref{corollary:bootstrap_stability}.
  Consider some configuration $c\in \blocking$. As the $2$-neighbour contamination $F$ is monotone and freezing, there exists a limit configuration $c^\infty = \lim\limits_{n\to \infty}F^n(c)$ and since $c \in \blocking \subset \invasion^\complement$, $c^\infty$ is not the $1$-uniform configuration, \emph{i.e.}, $c^\infty\neq 1^\tiling$.

  Recall that because $F$ is freezing, we have that for any tile $t\in T$, $c_t \leq c^\infty_t$.
  So we also have $c^\infty \in \blocking$.

  As $c^\infty$ is stable for \bootstrapCA{} and not the $1$-uniform configuration, any $1$ in $c^\infty$ is in a stable contaminated cluster.
  Corollary~\ref{corollary:bootstrap_stability} states that any stable contaminated cluster is either infinite or enclosed in a chain $2d$-gon of zeros.
  The case of an infinite contaminated cluster $P \subsetneq T$ is impossible, as the exterior boundary of $P$ would contain an infinite half-chain of zeros which is a contradiction with $c^\infty \in \blocking$.

  So any tile $t$ is either in state $0$ in $c$, in which case it is enclosed in the trivial chain $2d$-gon of zeros that is itself; or in state $1$ in $c$ in which case it is also in state $1$ in $c^\infty$ and therefore belongs to some finite stable cluster and is enclosed in a non-trivial chain $2d$-gon of zeros.
\end{proof}

\begin{definition}[Enclosed configurations]
  Let $\enclosed{t}$ be the set of configurations $x$ such that tile $t$ is enclosed in a closed chain $2d$-gon of $0$s, including the configurations where $x_t = 0$ as $t$ is enclosed in a trivial $2d$-gon.

  Let $\enclosed{t,\leq n}$ be the set of configurations $x$ such that $t$ is enclosed in a closed chain 2d-gon of $0$s of length at most $n$, and $\enclosed{t,>n}$ be $\enclosed{t}\setminus\enclosed{t,\leq n}$.
\end{definition}

By definition, for any $t$ and any $n$, we have $\enclosed{t,n} \subseteq \enclosed{t}$ and $\mu(\enclosed{t,n})\leq \mu(\enclosed{t})$. 
We will actually show that ${\mu(\enclosed{t})}$ is uniformly bounded away from $1$.
The first step is the following polynomial bound on the number of $2d$-gons.

\begin{lemma}[$2d$-gon counting]
  Given a rhombus tiling $T$ with finitely many edge directions.
  There exists a polynomial $Q$ such that for any tile $t$, the number of closed convex chain $2d$-gons around $t$ of length $n$ is at most $Q(n)$.
  \label{lemma:counting}
\end{lemma}
\begin{proof}
  Let us first remark that a chain $2d$-gon of length $n$ is in particular a $2d$-partition of $[n]$.

  Remark also that, when walking a convex chain $2d$-gon with exactly $2d$ sides in clockwise orientation, the order of the chain edge directions is fixed. If the $2d$-gon has less than $2d$ sides some directions are absent but the order of chain edge directions respects the fixed order. 

  This implies that given a tile $t$, and a starting point $p_0$ there are at most as many convex chain $2d$-gons of length $n$ arount $t$ and starting at $p_0$ as there are $2d$ partitions of $[n]$. In particular there are less than $n^{2d}$ such chain $2d$-gons. 

  Now remark that since $p_0$ is the starting point of a chain $2d$-gon of length $n$ that encloses $t$, the adjacency distance $d_T(p_0,t)$ is less than $n$.

  As the tiling $T$ has finitely many tiles up to translation, there is a maximal diameter of tiles $D$ and a minimal area of tiles $A$. So the euclidean distance from $p_0$ to $t$ is at most $D(n-1)$, and there are at most $\pi D^2n^2/A$ such starting points $p_0$ in $T$.

  Overall there are at most $Q(n):=\tfrac{\pi D^2}{A} n^{2d+2}$ convex chain $2d$-gons of length $n$ around $t$.

  This polynomial depends only on the tileset, that is, the set of tiles in $T$ up to translation. In particular $Q(n)$ is of degree $2d+2$ where $d$ is the number of edge directions in $T$, and the leading coefficient depends on the shapes of the tiles (minimal area and maximal diameter).
\end{proof}

\begin{lemma}[Uniform bound]
  There exists a $\lambda<1$ such that for any tile $t\in T$, $\mu(\enclosed{t})\leq \lambda < 1$.
  \label{lemma:uniform_bound}
\end{lemma}

\begin{proof}
  This lemma is a consequence of Lemma \ref{lemma:counting}. The key element is that there are polynomially many convex $2d$-gons of length $n$, but having $0$s on the whole $2d$-gon has exponential probabilistic cost.

  Consider the series $\sum\limits_{n\geq 0} Q(n)\beta^n$ where ${0<\beta<1}$ is the constant given by Lemma~\ref{lem:expdecay} on measure $\mu$. As $\beta<1$, and $Q$ is a polynomial the series converges. This implies that there exists $N\in \mathbb{N}$ such that $\sum\limits_{n\geq N}Q(n)\beta^n < 1$. Denote $\lambda_0 := \sum\limits_{n\geq N}Q(n)\beta^n$.

  Now consider a tile $t\in T$ and the patch ${P=B(t,N)}$ made of all tiles around $t$ up to distance $N$.
  Since $\mu$ is non-vanishing, there is some constant $\alpha>0$ depending on $N$ but not on $t$ such that ${\mu([1^P])\geq\alpha}$.
  
  \begin{adjustwidth}{1cm}{} 
    \begin{claim} 
      ${\mu([1^P] \cap \enclosed{t}^\complement)\geq \alpha(1-\lambda_0)}$
    \end{claim}
    \begin{proof}[Proof of the claim]
      Denote $\enclosed{P}$ the set of configuration where $P$ is enclosed in a convex chain $2d$-gon of zeros (not intersecting $P$), and $\enclosed{P,=n}$ the subset where it is enclosed in a convex chain $2d$-gon of length exactly $n$.
      
      Remark that $[1^P]\cap \enclosed{t}^\complement = [1^P]\cap \enclosed{P}^\complement$ since any chain $2d$-gon of zeros around $t$ in a configuration of $[1^P]$ necessarily avoids and encloses $P$.
      
      Remark also that if $n<N$, $\enclosed{P,=n} = \emptyset$ and for $n\geq N$ we have $\mu(\enclosed{P,=n})\leq Q(n)\beta^n$, by the union bound, as there are at most $Q(n)$ convex $2d$-gons of length $n$ around $P$ and $\enclosed{P,=n}$ is the union on all these $2d$-gons $\polygon$ of the event $[0^\polygon]$ which has measure less than $\beta^n$ by Lemma~\ref{lem:expdecay}.
      With $\enclosed{P} = \bigcup\limits_{n\in \mathbb{N}} \enclosed{P,=n}$ we have $\mu(\enclosed{P}) \leq \lambda_0 = \sum\limits_{n\geq N} Q(n)\beta^n$.

Both $[1^P]$ and $\enclosed{P}^\complement$ are upward-closed events so by hypothesis on $\mu$ they are positively correlated. Therefore we have: $\mu([1^P]\cap \enclosed{t}^\complement) \geq \alpha(1-\lambda_0)$ which proves the claim.
    \end{proof}
    \end{adjustwidth}
From the claim we deduce that $\mu(\enclosed{t}^\complement) \geq \alpha(1-\lambda_0)$.
Denote $\lambda := 1- \alpha(1-\lambda_0)$.
We have $\mu(\enclosed{t})\leq \lambda < 1$.
    \end{proof}

\begin{lemma}[Uniform approximation]
  For any $\epsilon$,
  there exists $n$ such that for any $t\in T$ we have
  \[\mu(\enclosed{t,> n}) \leq \epsilon\]
  \label{lemma:uniform_approximation}
\end{lemma}

\begin{proof}
  The proof is similar to the proof of Lemma \ref{lemma:uniform_bound} and relies on the fact that there are polynomially many convex chain $2d$-gons of length $n$.

  Denote $\enclosed{t,>n}'$ the set of configurations where tile $t$ is enclosed in a convex chain $2d$-gon of zeros of length more than $n$ (but possibly also in one of length less than $n$).

  By definition we have $\enclosed{t,>n} \subset \enclosed{t,>n}'$.

  Additionally $\enclosed{t,>n}'$ is the disjunction over all convex chain $2d$-gons of length more than $n$ of the event ``being all $0$'' so, denoting $\beta$ the constant from Lemma~\ref{lem:expdecay}, we have $\mu(\enclosed{t,>n}') \leq \sum\limits_{k>n} Q(k)\beta^k$.

  As there the series converge, there exists $n$ such that $\sum\limits_{k>n} Q(k)\beta^k<\epsilon$.

  We then have $\mu(\enclosed{t,>n})\leq \mu(\enclosed{t,>n}') \leq \epsilon$ as expected.
\end{proof}

\begin{lemma}
\label{lemma:mublockingzero}
  $\mu(\blocking) = 0$
\end{lemma}

\begin{proof}
  We prove that for any $m\in \mathbb{N}$ and any $\epsilon > 0$, we have $\mu(B)\leq \lambda^m + m\epsilon$ with $\lambda<1$ of Lemma \ref{lemma:uniform_bound}.

  Take $\epsilon > 0$.
  Take $m\in\mathbb{N}$.
  By Lemma \ref{lemma:uniform_approximation}, there exists $n$ such that for any tile $t$,
  $\mu(\enclosed{t,>n}) \leq \epsilon$.

  Take $m$ tiles $t_0\dots t_{m-1}$ sufficiently far away from each other, in such a way that no two chain $2d$-gons of length $n$ enclosing two distinct tiles can be at distance $k$ or less, where $k$ is such that $\mu$ is $k$-Markov.
  This means that the events $\enclosed{t_i,\leq n}$ are independent and therefore we have $\mu\left(\bigcap\limits_{0\leq i <m} \enclosed{t_i,\leq n}\right) = \prod\limits_{0\leq i <m} \mu(\enclosed{t_i,\leq n})$.

  We have
  \[ \blocking \subseteq \bigcap\limits_{0\leq i < m} \enclosed{t_i} \]
  due to the fact that for any configuration $x$ in $B$, for any $i \in \{ 0, \dots, m-1 \}$, either $x_{t_i} = 1$ and thus $x$ belongs to $E_{t_i}$ due to Lemma~\ref{lemma:blockingenclosed}, or $x_{t_i}=0$ and $t_i$ is enclosed in the trivial 2d-gon of $0$s of itself.
  
  Therefore, we have:
  \[ \blocking \subseteq \bigcap\limits_{0\leq i < m} \enclosed{t_i} \subseteq \bigcap\limits_{0\leq i <m} \left( \enclosed{t_i,\leq n} \cup  \left(\enclosed{t_i} \setminus \enclosed{t_i,\leq n}\right)\right) \subseteq \bigcap\limits_{0\leq i <m} \enclosed{t_i, \leq n} \cup \bigcup\limits_{0\leq i <m} \enclosed{t_i,>n}\]

  From this we get:
  \begin{align*}
    \mu(\blocking) &\leq \mu(\bigcap\limits_{0\leq i <m} \enclosed{t_i,\leq n}) + \mu\left( \bigcup\limits_{0\leq i <m} \enclosed{t_i,>n}\right) \\
    &\leq \mu(\bigcap\limits_{0\leq i <m} \enclosed{t_i,\leq n}) + \sum\limits_{0\leq i < m}\mu(\enclosed{t_i,>n})\\
    &\leq \mu(\bigcap\limits_{0\leq i <m} \enclosed{t_i,\leq n}) + m\epsilon\\
    &\leq \prod\limits_{0\leq i <m} \mu(\enclosed{t_i,\leq n}) + m\epsilon\\
    &\leq \prod\limits_{0\leq i < m} \mu(\enclosed{t_i}) + m\epsilon\\
    &\leq \lambda^m + m\epsilon
  \end{align*}

  Since this holds for any $m\in \mathbb{N}$ and any $\epsilon>0$, we get $\mu(\blocking)=0$.
\end{proof}

We call the proof technique above the \emph{polka dot technique}, since the $\enclosed{t_i,\leq n}$s focus on the surroundings of tiles far away from each other and therefore with independant behavior, just as polkka dots on our tiling.

\begin{theorem}[2-neighbour percolation on rhombus tilings]
  Let $\tiling$ be a rhombus tilings with finitely many edge directions.
  Let $F$ be the $2$-neighbour contamination cellular automaton on configurations $\{0,1\}^\tiling$.
  Let $I$ be the invasion set, that is $I:=\{ c \in \{0,1\}^\tiling \mid F^n(c) \to 1^\tiling\}$.
  Let $\mu$ be a \MNVPC{} measure on $\{0,1\}^\tiling$.
  We have \[ \mu(\invasion) = 1.\]
\end{theorem}
Recall in particular that non trivial Bernoulli measures are \MNVPC{} so this theorem holds for non trivial Bernoulli measures.

\begin{proof}
  Since $\mu(\chainzero)=0$ by Lemma~\ref{lemma:muchainzero}, we have $\mu(\chainzero^\complement)=1$. Therefore $\mu(\invasion^\complement) = \mu(\invasion^\complement \cap \chainzero^\complement) = \mu(\blocking)$.
  By Lemma \ref{lemma:mublockingzero}, $\mu(\blocking)=0$ so $\mu(\invasion^\complement) = 0$ and $\mu(\invasion)=1$.
\end{proof}

\section{Open questions}
This work opens more questions than it closes. We classify these open questions into three classes.

As mentioned in Section \ref{sec:bootstrap-intro}, the classical $0-1$ laws do not apply for dynamical or bootstrap percolation on rhombus tilings.
In all generality, for any rhombus tilings and any percolation process, the $0-1$ law does not hold, see for example Appendix \ref{appendix:perco-rhombus}.
However, we conjecture that if the rhombus tiling is sufficiently regular then the $0-1$ law holds for the invasion event of percolation processes.

\begin{conjecture}[$0-1$ law for uniformly repetitive rhombus tilings]
  Let $G=(V,E)$ be the adjacency graph of an uniformly repetitive (or uniformly recurrent) rhombus tiling.
  Let $F$ be a percolation process (monotone and freezing) on $\{0,1\}^V$.
  Let $I\subset \{0,1\}^V$ be the set of invading configurations for $F$.
  Let $\mu$ be a Bernoulli measure.
  \[\mu(I)\in\{0,1\}.\]
\end{conjecture}

On $\mathbb{Z}^2$, percolation processes have been studied in great generality, and though the exact critical threshold can often not be determined exactly, a trichotomy between trivial percolation threshold ($p_c\in \{0,1\}$) and non-trivial percolation threshold ($p_c>0$) can be determined from the stable directions of the percolation process~\cite{bollobas2015,balister2016,bollobas2023}.
We similarly conjecture that the non-triviality of the critical threshold on sufficiently regular rhombus tilings is determined by the stable directions.
We consider here the class of multigrid dual tilings \cite{debruijn1981,debruijn1986}, which are the canonical case of cut-and-project rhombus tilings: those are very regular, and in particular chains of rhombuses are almost straight (each chain of rhombuses is included in a tube of direction $\vec{e}^\bot$ where $\vec{e}$ is the edge direction of the chain). In particular the classical Penrose rhombus tilings and Ammann-Beenker rhombus tilings are multigrid dual tilings.

\begin{question}[Stable directions and critical probability on rhombus tilings]
  Let $\tiling$ be a multigrid dual tiling and $F$ a percolation process on $\{0,1\}^\tiling$.
  Do the stable directions of $F$ on $\tiling$ determine the trichotomy between trivial percolations thresholds and non-trivial percolation threshold?
\end{question}

The main result of the present article is that 2-neighbour percolation invades almost surely for any \MNVPC{} or non-trivial Bernoulli measure on any rhombus tilings, or equivalently on the adjacency graph of any rhombus tiling.
The first key element of the proof is the absence of fortress, the second is a counting argument on the possible finite walls of $0$s by bounding the number of wall directions.
It can easily be seen that for arbitrary $4$-regular graphs, even planar graph, if the graph has exponential growth (for example the Cayley graph of the free group on 2 generators) then the counting argument fails.
However, if the graph is quasi-isometric to $\mathbb{Z}^2$, the counting argument may hold without the strong structure of rhombus tilings.

\begin{question}[Almost sure percolation for 2-neighbour percolation on 4-regular graphs QI to $\ZZ^2$ without fortress]
  Let $G=(V,E)$ be a 4-regular graph that is quasi-isometric to $\mathbb{Z}^2$.
  Let $I\subset \{0,1\}^V$ be the set of invading configurations for 2-neighbour percolation on $G$.
  Let $\mu$ be a non trivial Bernoulli measure on $\{0,1\}^V$.
  
  If $G$ contains no fortress (finite obstacle) for 2-neighbour contamination, is it the case that $\mu(I)=1$?
\end{question}

2-neighbour percolation is very classicaly studied on $\mathbb{Z}^2$ and has been studied on trees and non-amenable groups \cite{balogh2006}. We aim to generalise this approach to other 2-generator groups such as Baumslagg Solitar groups.
\begin{question}[Cayley graph]
  Can we classify 2-generators finitely presented groups with regards to the critical probability of 2-neighbour percolation?
\end{question}

An auxiliary question that arose from our first exploration of percolation on Cayley graph is related to mixed degree graphs.
Consider for example $G$ a single sheet of the Cayley graph of $BS(1,2)$.
$G$ is planar and contains vertices of degree $4$ and $3$, each vertex of degree $3$ being adjacent to only vertices of degree $4$.
In a $2$-neighbour percolation on $G$, vertices of degree 3 are ``hard'' to contaminate as they require 2 out of 3 neighbours to be contaminated.
It appears that $G$ has a critical probability of percolation of $1$, that is for any non-trivial Bernoulli measure $\mu$, $\mu(I)=0$.
This is due to the uncountable number of vertical chains (each chain branching in two at each tile), linked to the fact that $G$ is not quasi-isometric to $\mathbb{Z}^2$.

However we may ask a similar question as above for mixed degree graphs that are quasi-isometric to $\ZZ^2$.

\begin{question}[Mixed degree]
  Let $G=(V,E)$ a mixed 3-4 degree graph quasi isometric to $\mathbb{Z}^2$ and such that degree 3 vertices are only adjacent to degree $4$ vertices.
  Let $I\subset \{0,1\}^V$ be the set of invading configurations for 2-neighbour percolation on $G$.
  Let $\mu$ be a non trivial Bernoulli measure on $\{0,1\}^V$.
  
  If $G$ contains no fortress (finite obstacle) for 2-neighbour contamination, is it the case that $\mu(I)=1$?
\end{question}

Note also that for the example of $G$ being a single sheet of the Cayley graph of $BS(1,2)$, there exists a second approach which consists in studying percolation not on vertices but on elementary cycles, or equivalently on the dual graph $G^*$.
In other words, we can consider $G$ either as a graph on which we study percolation, or as a tiling by elementary cycles for which we study the percolation on the adjacency or dual graph. 

\appendix
\section{The non $0-1$ percolation on a quadrilater tiling}
\label{appendix:nonzeroone}
Let $P_f$ be the fortress consisting of $4$ trapezes and a small square forming a big square presented in Fig.~\ref{fig:fortress}.
Let $\tiling_f$ be the square grid where the origin square has been replaced by $P_f$.

For a configuration $c\in \{0,1\}^{\tiling_f}$, we write $c|_{P_f}=0$ (resp. $c|_{P_f}=1$) when all the tiles of the fortress $P_f$ are in state $0$ (resp. $1$), $c|_{P_f}\notin\{0,1\}$ when at least a fortress tile is in state $0$ and at least one is in state $1$. For other tiles we denote $c(i,j)$ with $(i,j)\neq (0,0)$ for the state of the $(i,j)$ tile in the induced grid.

We write $F$ the cellular automaton of 2-neighbour contamination on both $\tiling_f$ and $\ZZ^2$.

\begin{proposition}
  \label{prop:bootstrapnonzeroone}
  Let $\mathfrak{F}$ be the set of configurations in $\{0,1\}^{\tiling_f}$where at least one tile of the fortress $P_f$ is in state $1$.
  Let $\invasion$ be the set of invading configurations.
  Let $\mu$ be the Bernoulli measure of parameter $p$.

  \[ \mu(I) = \mu(\mathfrak{F}) = 1-(1-p)^5\]
\end{proposition}
We give here an outline of the proof which is adapted from the strategy on rhombus tilings detailed in Section \ref{sec:mainresult}.

Denote $G_\circ$ the square grid with a single hole on the origin, that is $G_\circ := \mathbb{Z}^2\setminus \{(0,0)\}$.
Denote $F$ the $2$-neighbour contamination also on $G_\circ$, and $\invasion_\circ$ the set of invading configurations on $\{0,1\}^{G_\circ}$.

\begin{proposition}[Bootstrap percolation on the square grid with a hole]
	\label{prop:bootstrapwithhole}
  For any non-trivial Bernoulli measure $\mu$ on $\{0,1\}^{G_\circ}$, $\mu(I_\circ)=1$.
\end{proposition}

In what follows, we sketch the proof that the behavior of the square grid with a hole can be assimilated to the one of the full square grid $\ZZ^2$, and reuse our results on rhombus tilings to prove the almost-sure invasion.

Given a stable connected subset $P$ of $G_\circ$, we denote $\tilde{P}$ its induced $\mathbb{Z}^2$ subset by $\tilde{P}:= P\cup \{(0,0)\}$ if $P$ contains at least three of $\{(1,0), (-1,0),(0,1), (0,-1)\}$, and $\tilde{P}:= P$ otherwise.

We say that $P$ is regular if $|P\cap \{(1,0), (-1,0), (0,1), (0,-1)\}| \neq 2$ and singular otherwise.

If $P$ is singular, then exactly two of $\{(1,0), (-1,0), (0,1), (0,-1)\}$ are in $P$. Denote $H_1$ and $H_2$ the two half-planes among $\{(i,j)| i>0\}, \{(i,j) | i<0\}, \{(i,j) | j>0\}, \{(i,j) | j<0\}$ containing this two points.
Then denote $\tilde{P_1}:= P \cap H_1$ and $\tilde{P_2}:= P\cap H_2$.  Remark that because outside of position $(0,0)$, $\tilde{P}$ is stable for 2-neighbour contamination in $\mathbb{Z}^2$, we obtain that  $\tilde{P}=\tilde{P_1}\cup \tilde{P_2}$.

\begin{lemma}[Geometry and combinatorics of stable contaminated clusters on $G_\circ$]
  Let $P$ be a connected subset of $G_\circ$ that is stable for 2-neighbour contamination.

  If $P$ is regular, then $\tilde{P}$ is stable for $2$-neighbour contamination in $\mathbb{Z}^2$.

  If $P$ is singular, then both $\tilde{P}_1$ and $\tilde{P}_2$ are stable for $2$-neighbour contamination in $\mathbb{Z}^2$.
\end{lemma}
\begin{proof}[Sketch of proof]
  Remark that Lemma \ref{lemma:technical_stablechain} holds in $G_\circ$, though in $G_\circ$ the horizontal line (resp. vertical line) going through $0$ is actually two disjoint half-chains.
  The conclusion follows.
\end{proof}

From this remark on $\tilde{P}$, we obtain immediatly the following corollary on the shape of stable clusters.
\begin{corollary}
  Let $P$ be a  subset of $G_\circ$ that is stable for 2-neighbour contamination.
  If $P$ is finite, then $\tilde{P}$ is enclosed in either a rectangle or a L-shaped hexagon.
  If $P$ is infinite, then $\tilde{P}$'s boundary contains an infinite half-chain.
\end{corollary}

Here we call L-shaped hexagon the union of two rectangles of the form $([0,n_1]\times [0,m_1]) \cup ([0,n_2]\times [0,m_2])$, see Fig.~\ref{fig:possible_Gcirc_stable}.

\begin{figure}[htp]
  \center \includegraphics[width=0.8\textwidth]{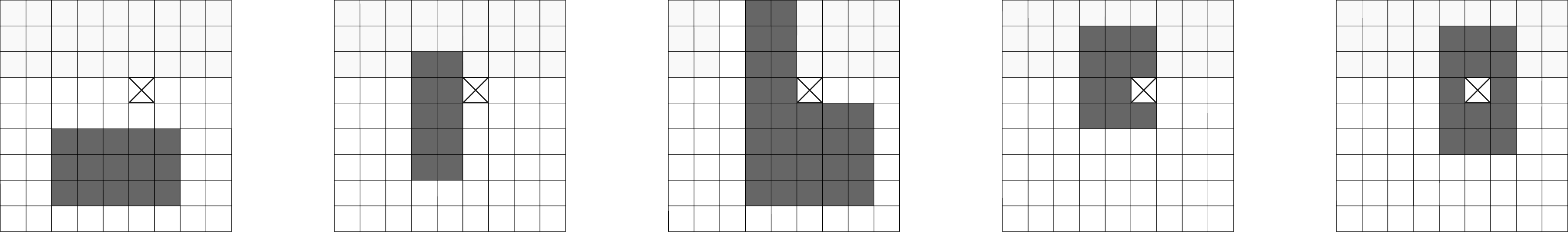}
  \caption{Possible stable clusters in the $G_\circ$, the crossed out cell is at position $(0,0)$.\\
  The middle one is the only singular stable cluster, it has a L-shaped hexagon shape.}
  \label{fig:possible_Gcirc_stable}

\end{figure}

As the number of ``sides'' of a ``wall'' around a finite contaminated cluster is bounded, then we obtain a polynomial upper bound on the number of possible walls of a given length as stated below.

\begin{corollary}[Counting possible finite clusters of a given perimeter]
  There is a polynomial $Q$ such that the number of possible walls of length $n$ around any given cell $t$ is at most $Q(n)$.
\end{corollary}

With these intermediate statements, we can now apply the proof strategy from Section \ref{sec:probargs} as outlined below

\begin{proof}[Outline of proof of Proposition \ref{prop:bootstrapwithhole} of almost sure invasion on $G_\circ$]
We can now apply the same techniques as for the rhombus case.
The set $\chainzero$ of configurations containing at least an infinite half-line or half-column in state $0$ is of measure $0$.

Define $E_t$ as the set of configurations where the tile $t$ is enclosed in a finite wall.

Denote $\blocking:= \invasion^\complement \cap \chainzero^\complement$.

We have $\blocking \subset \bigcap\limits_{t\in G_\circ} E_t$.

Now we have a uniform bound $\lambda<1$ on the measure of $E_t$, and a uniform approximation on $E_{t,>n}$.

We apply the polka dot technique of Lemma \ref{lemma:mublockingzero} to obtain $\mu(\blocking)=0$, which implies $\mu(\invasion_\circ^\complement)=0$ and $\mu(\invasion_\circ)=1$.
\end{proof}

This result implies that, on $\tiling_f$, almost surely all non-fortress cells eventually get contaminated.
As a consequence, full contamination happens almost surely as long as the fortress is initially disarmed (meaning it initially contains at least one contaminated cell), indeed a disarmed fortress gets contaminated if it is surronded.

Note that for a \MNVPC{} measure, we similarly have $\mu(\invasion)=\mu(\mathfrak{F})$ by a very similar proof.

Note that the proof outlined here also works for any single finite rectangular (or sufficiently well behaved)  hole (or fortress) in a square grid.

\section{A percolation process on rhombus tilings without $0-1$-law}
\label{appendix:perco-rhombus}
In this appendix we quickly show an example of non uniformly repetitive rhombus tiling $\tiling$ and a specific percolation process $F$ such  that $\tiling$ contains a single fortress for $F$ leading to a non $0-1$ probability of percolation for any parameter $0<p<1$.

Here we consider rhombus tilings with 5 edge directions $\zeta_i$ for $i=0...4$ (as for Penrose tilings).
In this setting we consider the specific tiling $\tiling$ of Fig.~\ref{fig:appendix-perco-rhombus}.

On this tiling, we have ten possible adjacency directions for tiles which are $\pm \zeta_i^\bot$ for $i=0...4$.
We say that $t'$ is adjacent to $t$ in direction $a$ if $t\cap t' \equiv a^\bot$ and the vector $\vec{tt'}$ from the barycenter of $t$ to that of $t'$ has positive scalar product with $a$.

We call \emph{partially directed boostrap percolation} $F_A$ the percolation process $F$ where a tile gets contaminated if it has at least two contaminated neighbours in directions $A \subset \{\pm \zeta_i^\bot \mid 0\leq i < 5\}$.
We consider the specific case where $A_3 = \{\zeta_0^\bot, \zeta_1^\bot, \zeta_2^\bot,\pm \zeta_3^\bot, \pm \zeta_4^\bot\}$
and the percolation process $F_3$ associated with $A_3$, that is:

\[F_3(c)_t =
  \begin{cases}
    1 &\text{ if $c_t=1$ or $t$ has at least two neighbours $t'$ and $t''$ in}\\
    	& \text{directions from $A$ with $c_{t'}=1 \wedge c_{t''}=1$}\\
    0 &\text{ else,}
  \end{cases}
\]

This means a given tiles gets contaminated if two of its neighbours are contaminated in the previous step, but not counting neighbours in directions $\{-\zeta_0^\bot, -\zeta_1^\bot, -\zeta_2^\bot\}$.

\begin{figure}[htp]
  \center \includegraphics[width=0.6\textwidth]{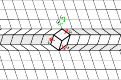}
  \caption{A rhombus tiling.
    Lines above and below the ones shown are made of the same rhombus tiles as the ones in white.\\
    The adjacency direction of partially directed bootstrap percolation are represented with the red and green arrows: green arrows are bidirectionnal contamination directions, red are directed contamination directions.\\
  The ``cube'' pattern with bold boundaries is a fortress as each tile only has one neighbour outside the pattern in a direction of contamination.}
  \label{fig:appendix-perco-rhombus}
\end{figure}

One can check that the only fortress for $F_3$ in $\tiling$ is the ``cube'' at the center of Figure \ref{fig:appendix-perco-rhombus} which is at the crossing of the $\zeta_0$, $\zeta_1$ and $\zeta_2$ lines (any other finite set of tiles $S$ contains a tile with two neighbours outside $S$ in directions from $A$).

Note that outside this singular band of tiles, in the white zones $F_3$ behaves as two neighbour percolations on $\mathbb{Z}\times \mathbb{N}$ as only directions $\zeta_3$ and $\zeta_4$ are present; and so almost surely the two half planes outside the band get invaded. This means that the right part of the band almost surely gets contaminated.
Indeed, if the outside half planes are fully $1$s and the grey chains on the right each contain a contaminated cell (say at positions $n_1$ and $n_2$, with position $0$ being the closest to the central ``cube'') then both grey chains get contaminated from $0$ to that contaminated position. Consequently the band gets fully contaminated from the fortress to the position $\min(n_1,n_2$).
Since for every $N$, there almost surely exist $n_1,n_2\geq N$ such that the grey cell of the top chain at position $n_1$ and of the top chain at position $n_2$ are contaminated, it follows that the whole right band almost surely gets contaminated.

So if the fortress has been disabled the entirety of the left part of the band gets contaminated, but if the fortress has not been disabled then the full invasion does not occur.

Overall the probability of invasion is exactly the probability that the fortress initally contains a contaminated tile; this probability is not in $\{0,1\}$ for non trivial Bernoulli distributions.

\printbibliography
\end{document}